%% file: ConcurrentSketches.tex
\documentclass[11pt, letterpaper]{article}

\input{environment/mypackages.tex}
\input{environment/mystyle.tex}

\begin{document}

\title{Fast Concurrent Data Sketches}

\author{Arik Rinberg \\ Technion \and Alexander Spiegelman \\ VMware Research \and Edward Bortnikov \\ Yahoo Research \and Eshcar Hillel
\\ Yahoo Research \and Idit Keidar \\ Technion \and Lee Rhodes \\ Verizon Media \and Hadar Serviansky \\ Weizmann Institute}









\maketitle

\input{sections/abstract.tex}



\input{sections/intro1.tex}

\input{sections/model1.tex}

\input{sections/background.tex}

\input{sections/ConcurrentSketches.tex}

\input{sections/GenericConcurrentAlgorithm.tex}

\input{sections/ErrorBounds.tex}

\input{sections/Evaluation1.tex}

\input{sections/discussion.tex}


\bibliographystyle{plain}

\bibliography{bibliography}



\appendix

\input{sections/artifact.tex}




\end{document}

%% file: environment/mypackages.tex
\usepackage{amsmath}
\usepackage{amssymb}
\usepackage{amsthm}
\usepackage{graphicx}
\usepackage{algorithm}
\usepackage[noend]{algpseudocode}
\usepackage{xcolor}
\usepackage{framed}
\usepackage{bm}
\usepackage{pifont}
\usepackage{multicol}
\usepackage{lipsum}
\usepackage{tikz}
\usetikzlibrary{calc}
\usepackage{wrapfig}
\usepackage{array}
\usepackage{mathtools}
\usepackage{url}
\usepackage{thm-restate}
\usepackage{varwidth}

\usepackage{caption}
\usepackage{subcaption}

%% file: environment/mystyle.tex
\newtheorem{claim}{Claim}

\newtheorem{definition}{Definition}

\newtheorem{lemma}{Lemma}

\newcommand{\ingray}[1]{{\color{gray}{#1}}}
\newcommand{\remove}[1]{}

\newcommand\Collection[1]{\{\,#1\,\}}
\newcommand{\abs}[1]{\lvert#1\rvert}

\DeclareMathOperator*{\argmax}{arg\,max}

\algblock{Vars}{EndFor}
\algrenewtext{Vars}{variables}
\algrenewtext{EndVars}{}
\algblock{ForEach}{EndFor}
\algrenewtext{ForEach}{for each }

%% file: sections/abstract.tex
\begin{abstract}
Data sketches are approximate succinct summaries of long data streams. They are widely used
for processing massive amounts of data and answering statistical queries about it.
Existing libraries producing sketches are very fast, but do not allow parallelism for 
creating sketches using multiple threads or querying them while
they are being built. We present a generic approach to parallelising data sketches efficiently and
allowing them to be queried in real time,
while bounding the error that such parallelism introduces. Utilising relaxed semantics and 
the notion of strong linearisability we prove our algorithm's correctness and analyse the 
error it induces in two specific sketches.
Our implementation achieves high scalability while keeping the error small. We have contributed
one of our concurrent sketches to the open-source data sketches library.
\end{abstract}

%% file: sections/intro1.tex
\section{Introduction}

Data sketching algorithms, or \emph{sketches} for short~\cite{Cormode:2017}, have become 
an indispensable tool for high-speed computations over massive datasets in recent years. 
Their applications include a variety of analytics and machine learning use cases, e.g., data aggregation~\cite{Agarwal:2012, KMV}, 
graph mining~\cite{Cohen:2014},  anomaly (e.g., intrusion) detection~\cite{Yang:2018}, real-time data analytics~\cite{DruidHLL},
and online classification~\cite{Tai:2018}.

Sketches are designed for \emph{stream} settings in which each data item is only processed once. A sketch data structure 
is essentially a succinct (sublinear) summary of a stream that approximates a specific query (unique element count, quantile values, etc.). 
The approximation is typically 
very accurate -- the error drops fast with the stream size~\cite{Cormode:2017}. 

Practical sketch implementations have recently emerged in toolkits~\cite{DataSketches}
and data analytics platforms (e.g., PowerDrill~\cite{GoogleHLL:2013}, Druid~\cite{DruidHLL}, Hillview~\cite{VMWareHillview}, and Presto~\cite{PrestoHLL}). 
However, these implementations are not thread-safe, allowing neither
parallel data ingestion nor concurrent queries and updates; concurrent use is prone to exceptions and 
gross estimation errors. Applications using these libraries are therefore required to explicitly protect all sketch API calls by locks~\cite{lee-groups-post, lee-issue}.

\begin{figure}[tb]
\setlength{\abovecaptionskip}{0pt}
\setlength{\belowcaptionskip}{0pt}
\setlength\textfloatsep{0pt}
  \begin{center}
    \includegraphics[width=\textwidth]{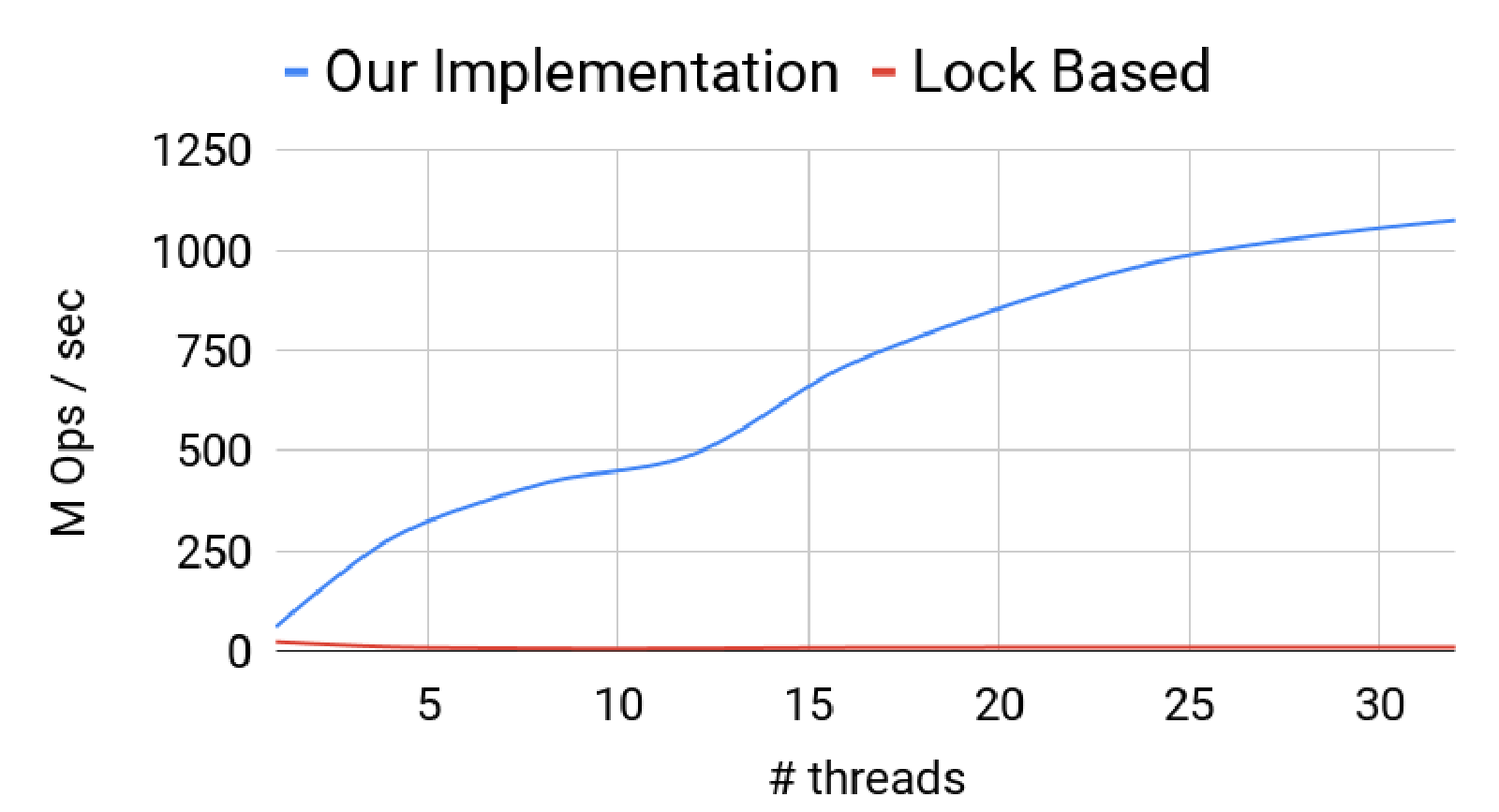}
  \end{center}
  \caption{Scalability of DataSketches' $\Theta$ sketch 
   protected by a lock vs. our concurrent implementation.}
  \label{fig:performance}
\end{figure}

We present a generic approach to parallelising data sketches efficiently while bounding the error that such a 
parallelisation might introduce. Our goal is to enable simultaneous queries and updates to a sketch from 
multiple  threads. Our solution is carefully designed to do so without slowing down operations as a result of synchronisation.
This is particularly challenging because sketch libraries are extremely fast, often processing tens of millions of updates per second. 

We capitalise on the well-known sketch \emph{mergeability} property~\cite{Cormode:2017}, which enables computing a sketch 
over a stream by merging sketches over substreams. Previous works have exploited this property for distributed 
stream processing (e.g.,~\cite{GoogleHLL:2013, cormode2011algorithms}), devising solutions with a sequential bottleneck at the merge phase and where queries cannot
be served before all updates complete. 
In contrast, our method is based on shared memory and constantly propagates results to a queryable sketch.
We adaptively parallelise  stream processing:  
for small streams, we forgo parallel ingestion as it might introduce significant errors;  
but as the stream becomes large, we process it in parallel using small
thread-local sketches with continuous background propagation of local results to the common (queryable) sketch.

We instantiate our generic algorithm with two popular sketches from the open-source Apache DataSketches 
(Incubating) library~\cite{DataSketches}:  
(1) a KMV $\Theta$ sketch~\cite{KMV}, which estimates the number of unique elements in a stream;
and (2) a Quantiles sketch~\cite{Agarwal:2012} estimating the stream element with a given rank. 
We have contributed the former back to the Apache DataSketches (Incubating) library~\cite{ConcurrentThetaImp}. 
Yet we emphasize that our design is generic and applicable to additional sketches. 

Figure~\ref{fig:performance} compares
the ingestion throughput of our concurrent $\Theta$ sketch to that of a lock-protected sequential sketch,
on multi-core hardware. As expected, the trivial solution does not scale whereas our algorithm scales linearly. 

Concurrency induces an error, and one of the main challenges we address is analysing this additional error.
To begin with, our concurrent sketch is a concurrent data structure, and 
we need to specify its  semantics. We do so using a flavour of 
\emph{relaxed consistency} due to Henzinger et al.~\cite{Henzinger}    
that allows operations to ``overtake'' some other operations.  
Thus, a query may return a result that reflects all but a bounded number of the updates
that precede it. 
While relaxed semantics were previously used for deterministic data structures like stacks~\cite{Henzinger}
and priority queues~\cite{alistarh, rihani2014multiqueues}, we believe that they are a natural fit for data sketches. 
This is because sketches are typically used to summarise streams that  arise from multiple real-world sources  
and are collected over a network with variable delays, and so even if the sketch ensures strict semantics, 
queries might miss some real-world events that occur before them. Additionally, sketches are inherently approximate.
Relaxing their semantics therefore ``makes sense'', as long as it does not excessively increase the expected error. 
If a stream is much longer than the relaxation bound, then indeed the error induced by the relaxation is negligible. But since the error allowed by such a relaxation is additive, in small streams, 
it may have a large impact. This motivates our adaptive solution, which forgoes relaxing small streams. 


We proceed to show that our algorithm satisfies relaxed consistency. 
But this raises a new difficulty: 
relaxed consistency is defined with regards to a deterministic
specification, whereas sketches are randomised. We therefore first de-randomise the sketch's behaviour
by delegating the random coin flips to an oracle. We can then relax the resulting sequential specification.
Next, because our concurrent sketch is used within randomised algorithms, 
it is not enough to prove its linearisability. Rather, 
we prove that our generic concurrent algorithm instantiated with sequential sketch $S$
satisfies \emph{strong linearisability}~\cite{Wojciech} with regards to a relaxed sequential specification of the de-randomised $S$. 
 
We then analyse the error of the relaxed sketches under random coin flips, with an adversarial scheduler that may delay operations in a
way that maximises the error. We show that our concurrent $\Theta$ sketch's error is coarsely bounded by twice
that of the corresponding sequential sketch. The error of the concurrent Quantiles sketch approaches
that of the sequential one as the stream size tends to infinity.

\vspace{-5pt}
\paragraph{Main contribution} In summary, this paper tackles an important practical problem,
offers a general efficient solution for it, and rigorously analyses this solution. While the
paper makes use of many known techniques, it combines them in a novel way. Alistarh et al.~\cite{alistarh2018distributionally}
present a formalisation for randomised relaxation of an object that has a sequential specification. Sketches have
randomised statistical guarantees, and as such do not have a sequential specification. Therefore, defining
the relaxed behaviour of these objects is non trivial.
The main technical challenges we address are (1) devising a high-performance generic algorithm 
that supports real-time queries concurrently with updates without inducing an excessive error; 
(2) proving the relaxed consistency of the algorithm; 
and (3) bounding the error induced by the relaxation in both short and long streams.

The paper proceeds as follows:
Section~\ref{sec:model} lays out the model for our work and Section~\ref{sec:background} provides background
on sequential sketches. In Section~\ref{sec:concurrentSketches} we formulate a flavour of relaxed semantics
appropriate for data sketches. Section~\ref{sec:genericAlg} presents our generic algorithm, and
Section~\ref{sec:error-bounds} analyses error bounds. Section~\ref{sec:eval} empirically studies the $\Theta$ sketch's performance
and error with different stream sizes. Finally, Section~\ref{sec:discussion}
concludes. The full paper~\cite{rinberg2019fast}
formally proves strong linearisability of our generic algorithm, and
includes some of the mathematical derivations used in our analysis.

%% file: sections/model1.tex
\section{Model}
\label{sec:model}

We consider a non-sequentially consistent shared memory model that enforces program order on all variables and allows  explicit 4
definition of \emph{atomic} variables as in Java~\cite{JavaMemoryModel} and C++~\cite{CppConcurrentMemoryModel}.
Practically speaking, reads and writes of atomic variables are guarded by memory fences, which guarantee
that all writes executed before a write {\sc w} to an atomic variable are visible to all
reads that follow (on any thread) a read {\sc r} of the same atomic variable s.t.\ {\sc r} occurs after {\sc w}. 

A thread takes \emph{steps} according to a deterministic \emph{algorithm} defined as a state machine. 
An \emph{execution} of an algorithm is an alternating sequence of steps and states, 
where  each step follows some thread's state machine.
Algorithms implement objects supporting \emph{operations}, such as query and update. 
An operation's execution consists of a series of steps, beginning with an \emph{invoke} and ending in a \emph{response}. 
The \emph{history} of an execution $\sigma$, denoted ${\mathcal{H}}(\sigma)$, 
is its subsequence of operation invoke and response steps.
In a \emph{sequential history}, each invocation is immediately followed by its response.
The \emph{sequential specification (SeqSpec)} of an object is its set of allowed sequential histories.

A \emph{linearisation} of a concurrent execution $\sigma$ is a history $H \in$\emph{SeqSpec}
such that (1) after adding responses to some pending invocations in $\sigma$ and removing others,
$H$ and $\sigma$ consist of the same invocations and responses (including parameters)
and (2) $H$ preserves the order between non-overlapping operations in $\sigma$.
Golab et al.~\cite{Wojciech} have shown that in order to ensure
correct behaviour of randomised algorithms under concurrency,
one has to prove \emph{strong linearisability}:

\begin{definition}[Strong linearisability]
A function $f$ mapping executions to  histories is \emph{prefix preserving} if
for every  two executions $\sigma, \sigma'$ s.t.\ $\sigma$ is a prefix of $\sigma'$,  
$f(\sigma)$ is a prefix of $f(\sigma')$.

An algorithm $A$ is a strongly linearisable implementation of an 
object $o$ if there is a prefix preserving function $f$ that maps 
every execution $\sigma$ of $A$ to a linearisation $H$ of $\sigma$.
\end{definition}

For example, executions of atomic variables are strongly linearisable.

%% file: sections/background.tex
\section{Background: sequential sketches}
\label{sec:background}

A sketch $S$ summarises a collection of elements \linebreak $\Collection{a_1,a_2,\dots,a_n}$, processed
in some order given as a stream $A=a_1,a_2,\dots,a_n$.
The desired summary is agnostic to the processing order,
but the underlying data structures may differ due to the order. Its API is:

\begin{description}
\item[$S$.init()] initialises $S$ to summarise the empty stream;
\item[$S$.update($a$)] processes stream element $a$;
\item[$S$.query($arg$)] returns the function estimated by the sketch over the stream processed thus far, e.g., the number of unique elements; 
 takes an optional argument, e.g., the requested quantile.
 \item[$S$.merge($S'$)] merges sketches $S$ and $S'$ into $S$; i.e., if $S$ initially summarised stream $A$ and $S'$ 
 summarised $A'$, then after this call, $S$ summarises the concatenation of the two, $A||A'$.
\end{description}

\paragraph{Example: $\Theta$ sketch}

Our running example is a $\Theta$ sketch based on the 
\emph{K Minimum Values (KMV)} algorithm~\cite{KMV} given in Algorithm~\ref{alg:composable-theta} (ignore the last
three functions for now). It maintains a \emph{sampleSet} and a parameter $\Theta$
that determines which elements are added to the sample set. 
It uses a random hash function $h$ whose outputs are uniformly distributed
in the range $[0,1]$, and $\Theta$ is always in the same range.  
An incoming stream element is first hashed, and then the hash is compared to $\Theta$. 
In case it is smaller, the value is added to \emph{sampleSet}.  Otherwise, it is ignored. 

Because the hash outputs are uniformly distributed, the expected proportion of values
smaller than $\Theta$ is $\Theta$. 
Therefore, we can estimate the number of unique elements in the stream by
dividing the number of (unique) stored samples by $\Theta$ (assuming that the random hash function is
drawn independently of the stream values).

KMV $\Theta$ sketches keep constant-size sample sets:
they take a parameter $k$ and keep the $k$ smallest hashes seen so far. 
$\Theta$ is $1$ during the first $k$ updates, and 
subsequently it is the hash of the largest sample in the set.
Once the sample set is full,
every update that inserts  a new element also removes the largest
one and updates $\Theta$. This is implemented efficiently using a min-heap.
The merge method adds a batch of samples to \emph{sampleSet}.

\paragraph{Accuracy}

Today, sketches are used sequentially,
so that the entire stream is processed 
and then $S$.query(arg) returns an estimate of the desired function 
on the entire stream. Accuracy is defined in one of two ways.
One approach analyses the \emph{Relative Standard Error (RSE)} of the estimate, 
formally defined in the full paper~\cite{rinberg2019fast},
which is the standard error normalized by the quantity being estimated.
For example, a KMV $\Theta$ sketch with $k$ samples has an RSE of less than $1/\sqrt{k-2}$~\cite{KMV}.

 A \emph{probably approximately correct (PAC)} sketch provides a result that estimates the correct result
within some error bound $\epsilon$ with a failure probability bounded by some parameter $\delta$.  
For example, a Quantiles sketch approximates the $\phi$th quantile of a stream with $n$ elements 
by returning an element whose rank is in $\left[(\phi-\epsilon)n , (\phi+\epsilon)n \right]$ with 
probability at least $1-\delta$~\cite{Agarwal:2012}.

%% file: sections/ConcurrentSketches.tex
\section{Relaxed consistency for concurrent sketches}
\label{sec:concurrentSketches}

We next relax sketch semantics to require that a query return a ``valid'' value reflecting some subset of the updates that have 
already been processed but not necessarily all of them.
We adopt a variant of Henzinger et al.'s~\cite{Henzinger} {\emph{out-of-order}} relaxation,  
which generalises quasi-linearisabilty~\cite{afek2010quasi}.
Intuitively, this relaxation allows a query to ``miss'' a bounded number of updates that precede it.
Because a sketch is order agnostic, we further allow re-ordering of the updates ``seen'' by a query.

A relaxed property for an object $o$ is an
extension of its sequential specification to allow more behaviours.
This requires $o$ to have a sequential specification, so
we convert sketches into deterministic objects by capturing their randomness in an external oracle; 
given the oracle's output, the sketches behave deterministically.
For the $\Theta$ sketch, the oracle's output is passed as a hidden variable to $init$, where the sketch
selects the hash function. In the Quantiles sketch, a coin flip is provided with every update.
For a derandomised sketch, we refer to the set of histories arising in its sequential
executions as \emph{SeqSketch}, and use SeqSketch as its sequential specification.
We can now define our relaxed semantics:
\begin{definition}[r-relaxation]
  A sequential history $H$ is an \emph{r-relaxation} of a sequential history $H'$,
  if $H$ is comprised of all but at most $r$ of the invocations in $H'$ and their responses,
  and each invocation in $H$ is preceded by all but at most $r$ of the invocations that precede the 
  same invocation in $H'$. The \emph{r-relaxation} of $SeqSketch$ is the set of histories
  that have r-relaxations in SeqSketch:
  
  $SeqSketch^r \triangleq $ $\{H'|\exists H\in$SeqSketch s.t. $H$ is an r-relaxation of $H'\}$.
  \label{def:r-relaxtion}
\end{definition}

Note that our formalism slightly differs from that of~\cite{Henzinger} in that we start with a serialisation $H'$ of an object’s
execution that does not meet the sequential specification and then ``fix'' it by relaxing it to a history $H$ in the sequential
specification. In other words, we relax history $H'$ by allowing up to $r$ updates to ``overtake'' every query, so the
resulting relaxation $H$ is in SeqSketch.
\begin{figure}[h]
    \begin{center}
      \includegraphics[width=\textwidth]{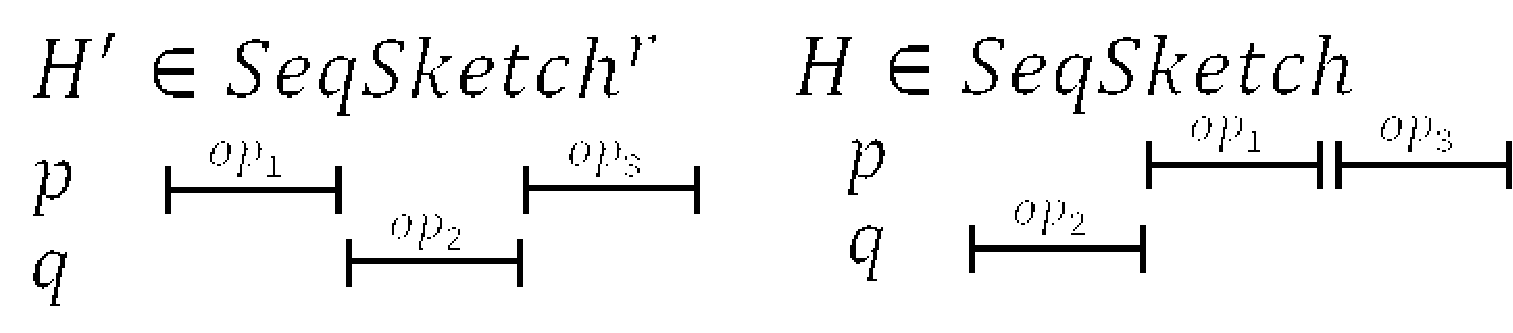}
    \end{center}
    \caption{$H$ is a 1-relaxation of $H'$.}
    \label{fig:relaxationExample}
\end{figure}

An example is given in Figure~\ref{fig:relaxationExample}, where $H$ is a 1-relaxation of history $H'$.
Both $H$ and $H'$ are sequential, as the operations don't overlap.

The impact of the $r$-relaxation on the sketch's error depends on the \emph{adversary}, which may select up to 
$r$ updates to hide from every query. There exist two adversary models:   
A \emph{weak adversary} decides which $r$ operations to omit from 
every query without observing the coin flips. 
A \emph{strong adversary} may select which updates to hide after learning 
the coin flips. Neither adversary sees the protocol's internal state, however both know the algorithm
and see the input. As the strong adversary knows the coin flips, it can then extrapolate the state; the
weak adversary, on the other hand, cannot.

%% file: sections/GenericConcurrentAlgorithm.tex
\section{Generic concurrent sketch algorithm}
\label{sec:genericAlg}

We now present our generic concurrent  algorithm. 
The algorithm uses, as a building block, an existing (non-parallel) sketch. 
To this end, we extend the standard sketch interface in Section~\ref{sec:composable-sketches}, 
making it usable within our generic framework.  
Our algorithm is adaptive -- it serialises ingestion in small streams and parallelises it in large ones.
For clarity of presentation, we  present in Section~\ref{sec:basic-generic-alg} the parallel phase of the algorithm, which provides relaxed semantics appropriate for large streams;  
in the full paper~\cite{rinberg2019fast}
we prove that it is strongly linearisable
with respect to an $r$-relaxation of the sequential sketch with which it is instantiated.
Section~\ref{ssec:small-streams} then discusses the adaptation for small streams.

\subsection{Composable sketches}
\label{sec:composable-sketches}

In order to be able to build upon an existing sketch \emph{S},
we first extend it to support a limited form of concurrency.
Sketches that support this extension are called \emph{composable}.

A composable sketch has to allow concurrency between merges and queries.
To this end, we add a \emph{snapshot} API that can run concurrently with merge and
obtains a queryable copy of the sketch. The sequential specification of this operation is as follows:
\begin{description}
    \item[$S$.snapshot()] returns a copy $S'$ of $S$ such that immediately after $S'$ is returned,
     $S$.query($arg$) = $S'$.query($arg$) for every possible $arg$.
\end{description}

A composable sketch needs to allow concurrency only between snapshots
and other snapshot and merge operations, and we require that such concurrent
executions be strongly linearisable. Our $\Theta$ sketch, shown below,
simply accesses an atomic variable that holds the query result. In other sketches
snapshots can be achieved efficiently by a double collect of the relevant state.

\paragraph{Pre-filtering} When multiple sketches are used in a multi-threaded algorithm,
we can optimise them by sharing ``hints'' about the processed data.
This is useful when the stream sketching function depends on the processed stream prefix.
For example, we explain below how $\Theta$ sketches sharing a common value of $\Theta$ can sample fewer updates.
Another example is reservoir sampling~\cite{vitter1985random}. To support this optimisation,
we add the following two APIs:
\begin{description}
    \item[$S$.calcHint()] returns a value $h \neq 0$ to be used as a hint.
    \item[$S$.shouldAdd($h$, $a$)] given a hint $h$, filters out updates 
    that do not affect the sketch's state.
\end{description}    
    Formally, the semantics of these APIs are defined using the notion of summary:
    (1) When a sketch is initialised, we say that its state (or simply the sketch) \emph{summarises} the empty history,
    and similarly, the empty stream; we refer to the sketch as \emph{empty}.
    (2) After we apply a sequential history 
		\[H= S.update(a_1), S.resp(a_1), \dots S.update(a_n), S.resp(a_n)\] 
    to a sketch $S$, we say that $S$ \emph{summarises} history $H$, and,
    similarly, summarises the stream $a_1, \dots ,a_n$.
    Given a sketch $S$ that summarises a stream $A$, if shouldAdd($S.calcHint()$, $a$) returns false then
    for every streams $B_1,B_2$ and sketch $S'$ that summarises $A||B_1||a||B_2$,
    $S'$ also summarises $A||B_1||B_2$.

These APIs do not need to support concurrency, and may be trivially implemented by always returning $true$.
Note that $S$.shouldAdd is a static function that does not depend on the current state of $S$.

\paragraph{Composable $\Theta$ sketch}
 
We add the three additional APIs to Algorithm~\ref{alg:composable-theta}.
The snapshot method copies $\mathit{est}$. Note that the
result of a merge is only visible after writing to est, because it is the only variable accessed by
the query. As $\mathit{est}$ is an atomic variable, the requirement on snapshot and merge is
met. To minimise the number of updates, calcHint returns $\Theta$
and shouldAdd checks if $h(a) < \Theta$, which is safe because the value of
$\Theta$ in sketch $S$ is monotonically decreasing. Therefore, if $h(a) \geq \Theta$
then $h(a)$ will never enter the \emph{sampleSet}.

\begin{algorithm}[htb]
    \small
    \begin{algorithmic}[1]
        \Vars
        \State   sampleSet, init $k$ $1$'s \Comment samples
        \State  $\Theta$, init $1$			\Comment threshold
        \State {\tt atomic} est, init $0$ \Comment estimate
        \State $h$, init random uniform hash function 
        \EndFor
        
        \Procedure{query}{arg}
        \State \Return $est$ \label{l:query}
        \EndProcedure
        
        \Procedure{update}{arg}
        \If{$h$(arg) $\geq \Theta$} \Return
        \EndIf
        \State add $h$(arg) to \emph{sampleSet}
        \State keep $k$ smallest samples in \emph{sampleSet}
        \State $\Theta \leftarrow max(sampleSet)$
        \State $\mathit{est}$ $\leftarrow $ $\left( |\text{sampleSet}|-1 \right)$ / $\Theta$
        \EndProcedure

        \Procedure{merge}{S}
        \State sampleSet $\leftarrow$ merge sampleSet and $S$.sampleSet
        \State keep $k$ smallest values in sampleSet
        \State $\Theta \leftarrow max($sampleSet$)$
        \State $\mathit{est}$ $\leftarrow $ $\left( |\text{sampleSet}|-1 \right)$ / $\Theta$ \label{l:update-est}
        \EndProcedure
        
        \Procedure{snapshot}{}
            \State $localCopy \leftarrow empty sketch$
            \State $localCopy.\mathit{est} \leftarrow \mathit{est}$
            \State \Return $localCopy$
        \EndProcedure
    
        \Procedure{calcHint}{}
            \State \Return $\Theta$
        \EndProcedure
    
        \Procedure{shouldAdd}{H, arg}
            \State \Return $h$(arg) $< H$
        \EndProcedure

    \end{algorithmic}
    \caption{Composable $\Theta$ sketch.}
    \label{alg:composable-theta}
\end{algorithm}

\subsection{Generic algorithm}
\label{sec:basic-generic-alg}

To simplify the presentation and proof, we first discuss an unoptimised version of
our generic concurrent algorithm (Algorithm~\ref{alg:generic-concurrent} without the gray lines)  called 
\emph{ParSketch},
and later an optimised version of the same algorithm (Algorithm~\ref{alg:generic-concurrent} including the gray lines
and excluding underscored line~\ref{l:signal}).

The algorithm is instantiated by a composable sketch and sequential sketches.
It uses multiple threads to process incoming stream elements 
and services queries at any time during the sketch's construction.
Specifically, it uses $N$ worker threads, $t_1,\dots,t_N$, each of which samples
stream elements into a local sketch $localS_i$, and a propagator thread $t_0$ that merges local sketches
into a shared composable sketch $globalS$. Although the local sketch resides in
shared memory, it is updated exclusively by its owner update thread $t_i$ and 
read exclusively by $t_0$. Moreover, updates and reads do not happen in
parallel, and so cache invalidations are minimised. The global sketch is updated only by $t_0$
and read by query threads. We allow an unbounded number of query threads. 

After $b$ updates are added to $localS_i$, $t_i$ signals to the propagator to merge
it with the shared sketch. It synchronises with $t_0$ using a 
single \emph{atomic} variable $prop_i$, which $t_i$ sets to 0.  
Because $prop_i$ is atomic, the memory model
guarantees that all preceding updates to $t_i$'s local sketch are visible to
the background thread once $prop_i$'s update is.
This signalling is relatively expensive (involving a memory fence),  
but we do it only once per $b$ items retained in the local sketch.

After signalling to $t_0$, $t_i$ waits
until $prop_i \neq 0$  (line~\ref{l:wait}); 
this indicates that the propagation has completed, and $t_i$ can 
reuse its local sketch. Thread $t_0$ piggybacks the hint \emph{H} it
obtains from the global sketch on $prop_i$,
and so there is no need for further synchronisation in order to pass the hint.

Before updating the local sketch, $t_i$ invokes shouldAdd to check
whether it needs to process \emph{a} or not. For example, the $\Theta$ sketch discards updates whose hashes are
greater than the current value of $\Theta$. The global thread passes the global sketch's
value of $\Theta$ to the update threads, pruning updates that would end up being discarded
during propagation. This significantly
reduces the frequency of propagations and associated memory fences.

\begin{algorithm}[htb]
    \small
    \begin{algorithmic}[1]
    \setcounter{ALG@line}{100}

    \Vars
    \State composable sketch \emph{globalS}, init empty
    \State constant $b$ \Comment relaxation is $2Nb$
    \ForEach{update thread $t_i$} , $0 \leq i \leq N$
        \State sketch \emph{localS$_i$}\ingray{[$2$]}, init empty
        \ingray{\State int \emph{cur}$_i$, init 0}
        \State int $counter_i$, init $0$
        \State int \emph{hint}$_i$, init $1$
        \State int {\tt atomic} $prop_i$, init $1$
    \EndFor
    \EndFor

    \Procedure{propagator}{}
    \While {true}
    \ForAll{thread $t_i$ s.t. $prop_i=0$}
        \State $globalS.merge(localS_i$\ingray{[1-$cur_i$]}$)$ \label{l:merge}
        \State $localS_i$\ingray{[1-$cur_i$]}$ \leftarrow $empty sketch \label{l:emptyAux}
		\State $prop_i \leftarrow globalS.calcHint()$ \label{l:calcHint}
    \EndFor
    \EndWhile
    \EndProcedure


    \Procedure{query}{arg}
    \State $localCopy \leftarrow globalS.snapshot(localCopy)$
    \State \Return $localCopy.query(arg)$
    \EndProcedure
    
    \Procedure{update$_i$}{$a$}
    \If{$\neg$shouldAdd(\emph{hint}$_i$, $a$)} \Return \label{l:shouldAdd}
    \EndIf
    \State $counter_i \leftarrow counter_i + 1$ \label{l:countup}
    \State $localS_i$\ingray{[$cur_i$]}$.update(a)$ \label{l:update}
    \If{$counter_i = b$} \label{l:checkfull}
    \State \underline{$prop_i \leftarrow 0$} \Comment In non-optimised version\label{l:signal}
    \State wait until $prop_i \neq 0$ \label{l:wait}
    \ingray{\State $cur_i \leftarrow 1 - cur_i$} \label{opt:l:swap-local-aux}
    \State \emph{hint}$_i \leftarrow prop_i$ \label{l:updateHint}
    \State $counter_i \leftarrow 0$ \label{l:zeroCounter}
    \ingray{\State $prop_i \leftarrow 0$} \Comment In optimised version\label{opt:l:signal}
    \EndIf
    \EndProcedure

    \end{algorithmic}
    \caption{\ingray{Optimised} generic concurrent algorithm.}
    \label{alg:generic-concurrent}
\end{algorithm}

Query threads use the snapshot method, which can be safely run concurrently with merge,
hence there is no need to synchronise between the query threads and $t_0$. The freshness
of the query is governed by the $r$-relaxation.
In the full paper~\cite{rinberg2019fast},
we prove Lemma~\ref{lemma:genereic-strong} below, asserting that
the relaxation is $Nb$. This may seem straightforward as $Nb$ is the combined size of the
local sketches. Nevertheless, proving this is not trivial because the local sketches pre-filter
many additional updates, which, as noted above, is instrumental for performance. 

\begin{lemma}
    $ParSketch$ instantiated with $SeqSketch$ is strongly linearisable with regards to $SeqSketch^{Nb}$.
    \label{lemma:genereic-strong}
\end{lemma}

A limitation of \emph{ParSketch} is that update threads are idle while waiting for the propagator to execute the merge. This
may be inefficient, especially if a single propagator iterates through many local sketches.
Algorithm~\ref{alg:generic-concurrent} with the gray lines included and the underlined line omitted presents
the optimised \emph{OptParSketch} algorithm, which improves thread utilisation via
double buffering.

In \emph{OptParSketch}, $localS_i$ is an array of two sketches. When $t_i$ is ready to propogate $localS_i[cur_i]$, it
flips the $cur_i$ bit denoting which sketch it is currently working on (line~\ref{opt:l:swap-local-aux}), 
and immediately sets $prop_i$ to 0 (line~\ref{opt:l:signal}) in order to allow the propagator to
take the information from the other one. It then starts digesting updates in a fresh sketch.

In the full paper~\cite{rinberg2019fast}
we prove the correctness of the optimised
algorithm by simulating $N$ threads of \emph{OptParSketch}
using $2N$ threads running \emph{ParSketch}. We do this by showing
a \emph{simulation relation}~\cite{lynch1996distributed}. We use forward simulation (with
no prophecy variables), ensuring strong linearisability. We conclude the following theorem:
\begin{restatable}{rthm}{optgenereicstrong}
    \emph{OptParSketch} instantiated with $SeqSketch$ is strongly linearisable with regards to \emph{SeqSketch}$^{2Nb}$.
    \label{lemma:opt-genereic-strong}
\end{restatable}

\subsection{Adapting to small streams}
\label{ssec:small-streams}

By Theorem~\ref{lemma:opt-genereic-strong}, a query can miss up to $r$ updates. For small
streams, the error induced by this can be very large.
For example, the sequential $\Theta$ sketch answers queries with perfect accuracy in streams with
up to $k$ unique elements, but if $k<r$, the relaxation can miss \emph{all} updates.
In other words, while the additive error is guaranteed to be bounded by $r$, the relative 
error can be infinite.  

To rectify this, we implement \emph{eager propagation} for small streams, 
whereby update threads propagate updates immediately to the shared sketch 
instead of buffering them. 
Note that during the eager phase, updates are processed sequentially. 
Support for eager propagation can be added to Algorithm~\ref{alg:generic-concurrent} 
by initialising $b$ to $1$ and having the propagator thread raise it to the
desired buffer size once the stream exceeds some pre-defined length. 
The error analysis of the next section can be used to determine the adaptation point.

%% file: sections/ErrorBounds.tex
\section{Deriving error bounds}
\label{sec:error-bounds}
Section~\ref{ssec:theta-analysis} discusses the error introduced to the
expected estimation and RSE of the KMV $\Theta$ sketch.
Section~\ref{ssec:quantiles-error-analysis} analyses the PAC Quantiles sketch.
The full paper~\cite{rinberg2019fast} contains mathematical derivations used throughout this section. 

\subsection{$\Theta$ error bounds}
\label{ssec:theta-analysis}

We bound the error introduced by an $r$-relaxation of the $\Theta$ sketch. Given
Theorem~\ref{lemma:opt-genereic-strong}, the optimised concurrent sketch's error is bounded
by the relaxation's error bound for $r=2N$$b$. We consider strong and weak adversaries,
${\mathcal{A}}_s$ and ${\mathcal{A}}_w$, resp. For the strong adversary we are able to show only numerical
results, whereas for the weak one we show closed-form bounds. The results are summarised in Table~\ref{table:Theta-Error-Summary}.
Our analysis relies on known results from order statistics~\cite{david2004order}.
It focuses on long streams, and assumes $n>k+r$.


\begin{table*}[!ht]
    \begin{tabular}{c|cc|cc|cc}
        & \multicolumn{2}{c|}{Sequential sketch} & \multicolumn{2}{c|}{Strong adversary ${\mathcal{A}}_s$} & \multicolumn{2}{c}{Weak adversary ${\mathcal{A}}_w$}   \\
        & Closed-form& Numerical& \multicolumn{2}{c|}{Numerical} & \multicolumn{2}{c}{Closed-form}   \\
        \hline
        Expectation & $n$        & $2^{15}$        & \multicolumn{2}{c|}{$2^{15} \cdot 0.995$}          & \multicolumn{2}{c}{$n\frac{k-1}{k+r-1}$} \\
        RSE & $\leq \frac{1}{\sqrt{k-2}}$        & $\leq 3.1\%$        & \multicolumn{2}{c|}{$\leq 3.8\%$}           & \multicolumn{2}{c}{$\leq 2\frac{1}{\sqrt{k-2}}$}         
    \end{tabular}
    \caption{Analysis of $\Theta$ sketch with numerical values for $r=8, k=2^{10}, n=2^{15}$.}
    \label{table:Theta-Error-Summary}
\end{table*}

We would like to analyse the distribution of the $k^{th}$ largest element in the 
stream that the relaxed sketch processes, as this determines the result returned by the algorithm. 
We cannot use order statistics to analyse this 
because the adversary alters the stream and so the stream seen by the algorithm is not random.
However, the stream of hashed unique elements seen by the adversary \emph{is} random. 
Furthermore, if the adversary hides from the algorithm $j$ elements 
smaller than $\Theta$, then the $k^{th}$ largest element in the stream seen
by the  sketch is the $(k+j)^{th}$ largest element in the original stream seen by the adversary. 
This element is a random variable and therefore we can apply order statistics to it.  

We thus model the hashed unique elements in the stream $A$ processed before a given
query as a set of $n$ labelled iid random variables $A_1,\dots,A_n$, taken uniformly 
from the interval $[0,1]$. Note that
$A$ is the stream observed by the reference sequential sketch, and 
also by adversary that hides up to $r$ elements from the relaxed sketch. 
Let $M_{(i)}$ be the $i^{th}$ minimum value among the $n$ random variables $A_1,\dots,A_n$.

Let $est(x) \triangleq \frac{k-1}{x}$ be the estimate computation
with a given $x=\Theta$ (line~\ref{l:update-est} of Algorithm~\ref{alg:composable-theta}).
The sequential (non-relaxed) sketch returns $e=est(M_{(k)})$.
It has been shown that the sketch is unbiased~\cite{KMV}, i.e., $E[e]=n$ the number of unique elements,
and $RSE[e]\leq \frac{1}{\sqrt{k-2}}$.
The \emph{Relative Standard Mean Error (RSME)} is the error relative to the mean, formally defined in
the full paper~\cite{rinberg2019fast}.
Because this sketch is unbiased, \emph{RSE}$[e]=$\emph{RSME}$[e]$.

In a relaxed history, the adversary chooses up to $r$ variables to hide from the given query so as to maximise its
error. It can also re-order elements, but the state of a $\Theta$ sketch after a set of updates
is independent of their processing order. Let $M^r_{(i)}$ be the $i^{th}$ minimum value among
the hashes seen by the query, i.e., arising in updates that precede the query in the relaxed history.
The value of $\Theta$ is $M^r_{(k)}$, which is equal to $M_{(k+j)}$ for some $0 \leq j \leq r$.
We do not know if the adversary can actually control $j$,
but we know that it can impact it, 
and so for our error analysis, we consider strictly stronger adversaries --
we allow both the weak and the strong adversaries to choose the number of hidden
elements $j$. Our error analysis gives an upper bound on the error induced by our adversaries.
Note that the strong adversary can choose $j$ based on the coin flips,
while the weak adversary cannot, and therefore chooses the same $j$ in all runs.
In the full paper~\cite{rinberg2019fast} we show that
the largest error is always obtained either for $j=0$ or for $j=r$. 
\begin{figure}[b]
    \begin{center}
        \includegraphics[width=\textwidth]{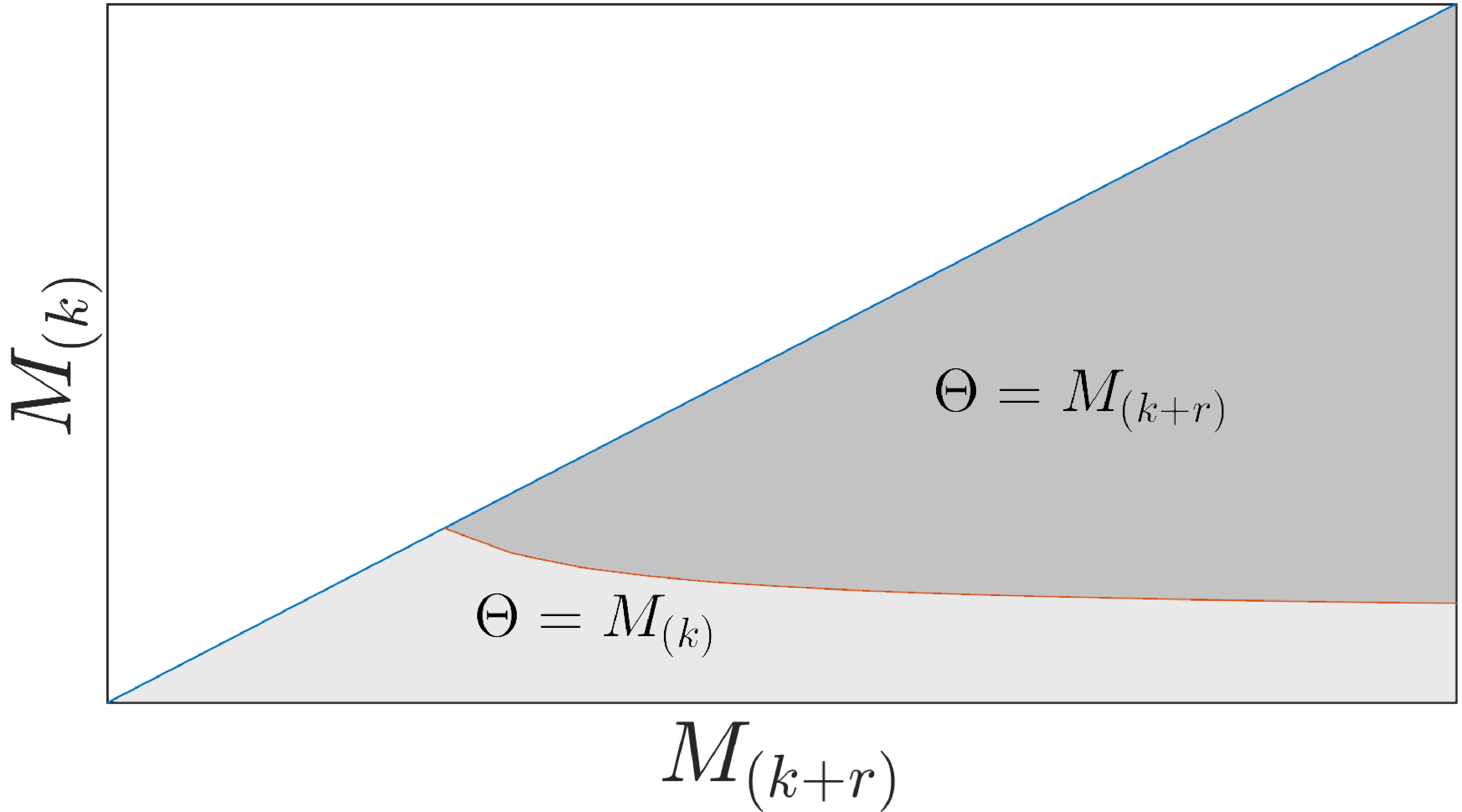}
    \end{center}
    \caption{Areas of $M_{(k)}$ and $M_{(k+r)}$. In the dark gray 
    ${\mathcal{A}}_s$ induces $\Theta=M_{(k+r)}$, and in the light gray, $\Theta=M_{(k)}$. The white
    area is not feasible.} 
    \label{fig:areaGraph}
\end{figure}

Given an adversary $\mathcal{A}$ that induces an approximation $e_{\mathcal{A}}$, in the full paper~\cite{rinberg2019fast} we prove
the following bound:
\begin{align*}
    \text{RSE}[e_{\mathcal{A}}] \leq \sqrt{\frac{\sigma^2(e_{\mathcal{A}})}{n^2}} + \sqrt{\frac{(E[e_{\mathcal{A}}] - n)^2}{n^2}}.
\end{align*}

\paragraph{Strong adversary ${\mathcal{A}}_s$} The strong adversary knows the coin flips in advance, and thus chooses
$j$ to be $g(0, r)$, where $g$ is the
choice that maximises the error:
\begin{align*}
    g(j_1, j_2) \triangleq \argmax_{j \in \left\{j_1, j_2\right\}} \abs{\frac{k-1}{M_{(k+j)}} - n}.
\end{align*} 

\begin{figure}[b]
    \begin{center}
        \includegraphics[width=\textwidth]{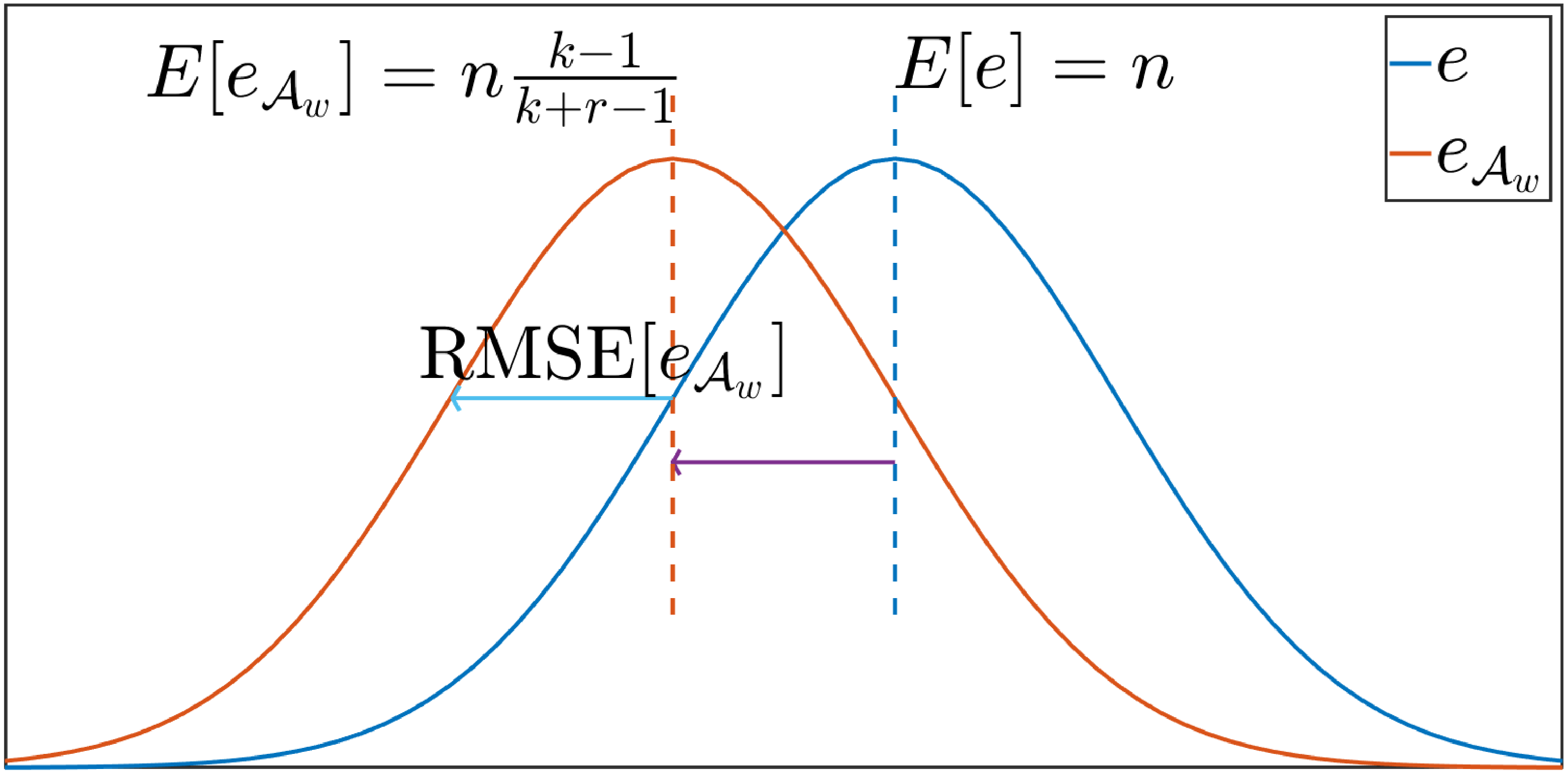}
    \end{center}
    \caption{Distribution of estimators $e$ and $e_{{\mathcal{A}}_w}$. The RSE of $e_{{\mathcal{A}}_w}$ with regards to $n$ is bounded
    by the relative bias plus the RMSE of $e_{{\mathcal{A}}_w}$.}
    \label{fig:thetaGraph}
\end{figure}

In Figure~\ref{fig:areaGraph} we plot the regions where $g$ equals $0$
and $g$ equals $r$, based on their possible combinations of values. The estimate
induced by ${\mathcal{A}}_s$ is $e_{{\mathcal{A}}_s} \triangleq \frac{k-1}{M_{(k+g(0,r))}}$. The expectation
and standard error of $e_{{\mathcal{A}}_s}$ are calculated by integrating over the gray areas
in Figure~\ref{fig:areaGraph}
using their joint probability function from order statistics. In the full paper~\cite{rinberg2019fast} we give
the formulas for the expected estimate and its RSE bound, resp. We do not have
closed-form bounds for these equations. Example numerical results are
shown in Table~\ref{table:Theta-Error-Summary}.

\paragraph{Weak adversary ${\mathcal{A}}_w$} Not knowing the coin flips, ${\mathcal{A}}_w$ chooses $j$
that maximises the expected error for a random hash function:
$E[n-est(M^r_{(k)})]=E[n-est(M_{(k+j)})]=n-n\frac{k-1}{k+j-1}$. Obviously this
is maximised for $j=r$. The orange curve in Figure~\ref{fig:thetaGraph} depicts
the distribution of $e_{{\mathcal{A}}_w}$, and the distribution of $e$ is shown in blue.



In the full paper~\cite{rinberg2019fast},
we show that the RSE bound of Lemma~\ref{lemma:theta-adversary-bound}
is bounded by $\sqrt{\frac{1}{k-2}} + \frac{r}{k-2}$ for $\hat{{\mathcal{A}}}_w$, and therefore so
is that of ${\mathcal{A}}_w$.
Thus, whenever $r$ is at most $\sqrt{k-2}$, the RSE of the relaxed
$\Theta$ sketch is coarsely bounded by
twice that of the sequential one. And in case $k \gg r$, the addition to the $RSE$ is negligible.

\subsection{Quantiles error bounds}
\label{ssec:quantiles-error-analysis}

We now analyse the error for any implementation of the sequential Quantiles sketch, provided that the sketch is
\emph{PAC}, meaning that a query for quantile $\phi$
returns an element whose rank is between $(\phi-\epsilon)n$ and $(\phi+\epsilon)n$ with 
probability at least $1-\delta$ for some parameters $\epsilon$ and $\delta$. We show that the $r$-relaxation of
such a sketch returns an element whose rank is in the range $(\phi \pm\epsilon_r)n$ with probability at
least $1-\delta$ for $\epsilon_r=\epsilon - \frac{r \epsilon}{n} + \frac{r}{n}$.


Although the desired summary is order agnostic here too, Quantiles sketch implementations (e.g., \cite{Agarwal:2012})
are sensitive to the processing order. In this case, advanced knowledge of the coin flips can increase the error
already in the sequential sketch. Therefore, we do not consider a strong adversary, but rather discuss only the weak one.
Note that the weak adversary attempts to maximise $\epsilon_r$.

Consider an adversary that knows $\phi$ and chooses to hide
$i$ elements below the $\phi$ quantile and $j$ elements above it, such that $0\leq i+j\leq r$. The rank of the element
returned by the query among the $n-(i+j)$ remaining elements is in the range 
$\phi(n-(i+j)) \pm \epsilon(n-(i+j))$.
There are $i$ elements below this quantile that are missed, and therefore its rank in the original stream is in the range:
\begin{equation}
    \left[ (\phi-\epsilon)(n-(i+j)) + i , (\phi+\epsilon)(n-(i+j)) + i \right].
    \label{eq:rank-range}
\end{equation}

This can be rewritten as:
\begin{equation}
\begin{split}
    [\phi n - (\phi j - (1-\phi)i+\epsilon(n-(i+j))), \\
    \phi n + ((1-\phi)i - \phi j +\epsilon(n-(i+j))) ] 
\end{split}
    \label{eq:rank-range-2}
\end{equation}

In the full paper~\cite{rinberg2019fast}, we
show that the $r$-relaxed sketch returns an element whose rank is
between $(\phi-\epsilon_r)n$ and $(\phi+\epsilon_r)n$ with probability at
least $1-\delta$, where $\epsilon_r=\epsilon - \frac{r \epsilon}{n} + \frac{r}{n}$. Thus
the impact of the relaxation diminishes as $n$ grows.

%% file: sections/Evaluation1.tex
\section{$\Theta$ sketch evaluation}
\label{sec:eval}

This section presents an evaluation of an implementation of our algorithm for the $\Theta$ sketch.
Section~\ref{ssec:setup-and-methodology} presents the methodology used to analyse the performance,
and the setup used for the analysis. Section~\ref{ssec:results} shows the results under different
workloads and scenarios. Finally, Section~\ref{ssec:tradeoffs} discusses the tradeoff between
accuracy and throughput.

\subsection{Setup and methodology}
\label{ssec:setup-and-methodology}

Our implementation extends the code in Apache DataSketches (Incubating)~\cite{DataSketches}, a Java
open-source library of stochastic streaming algorithms. The $\Theta$ sketch there differs slightly
from the KMV $\Theta$ sketch we used as a running example, and is based on a HeapQuickSelectSketch family.
In this version, the sketch stores between $k$ and $2k$ items, whilst keeping $\Theta$ as the $k^{\text{th}}$
largest value. When the sketch is full, it is sorted and the largest $k$ values are discarded.

Concurrent $\Theta$ sketch is generally available in the Apache DataSketches (Incubating)
library since V0.13.0. The sequential implementation and the sketch at the core of the global sketch
in the concurrent implementation are the same \\ (HeapQuickSelectSketch, which is the default sketch family).

As explained in Section~\ref{ssec:small-streams}, we implement a limit for eager propagation as a function
of the configurable error parameter $e$; the function we us is $2 / e^2$. The local sketches define $b$
as a function of $k$, $e$, and $N$ (the number of writer threads). The error induced by the relaxation
does not exceed $e$, and thus is a bound by $\max\{e + \frac{1}{\sqrt{k}}, \frac{2}{\sqrt{k}}\}$.

Eager propagation, as described in the pseudo-code, requires context switches incurring a high overhead. In the
implementation, either the local thread itself executes every update to the global sketch (equivalent to a
buffer size of 1) or lazily delegates updates to a background thread. While the sketch is in eager propagation
mode, the global sketch is protected by a shared boolean flag. When the sketch switches to estimate mode it
is guaranteed that no eager propagation gets through; instead local threads pass the buffer via lazy propagation.
This implementation ensures that: (a) local threads avoid costly context switches when the sketch is small, and (b)
lazy propagation by a background thread is done without synchronisation.

Unless otherwise stated, sketches are configured with $k=4096$, and $e=0.04$; thus the exact limit
is $2/e^2=1250$, and $b$ is set (by the implementation) to a value between $1$ and $5$ to accommodate the
error bound. Our first set of tests run on a 12-core Intel Xeon E5-2620 machine -- this machine is similar
to that which is used by production servers. For the scalability evaluation we used a 32-core Intel Xeon
E5-4650 to get a large number of threads. Both machines have hyper-threading disabled, as it introduces
non-monotonic effects among threads sharing a core.

We focus on two workloads: (1) write-only -- updating a sketch with a stream of unique values; (2) mixed
read-write workload -- updating a sketch with background reads querying the number of unique values in
the stream. Background reads refer to dedicated threads that occasionally (with $1$ms pauses) execute a query.
These workloads were chosen to simulate read-world scenarios where updates are constantly streaming from
a feed or multiple feeds, while queries arrive at a lower rate.

To run the experiments we employ a multi-thread extension of the characterization
framework. This is the Apache DataSketch evaluation benchmark suite, which measures
both the speed and accuracy of the sketch. 

For measuring write throughput, the sketch is fed with a continuous data stream. The size of
the stream varies from 1 to 8M unique values. For each size $x$ we measure the time $t$ it takes to feed the
sketch $x$ unique values, and present it in term of throughput ($x/t$). To minimise measurement noise,
each point on the graph represents an average of many trials. The number of trials is very high
($2^{18}$) for points at the low end of the graph. It gradually decreases as the size of the
sketch increases. At the high end (at 8M values per trial) the number of trials is 16. This is because
smaller stream sizes tend to suffer more from measurement noise.

Accuracy of concurrent $\Theta$ sketch is measured only in a single-thread environment. As
in the performance evaluations the $x$-axis represents the number of unique values fed into a sketch
by a single writing thread. For each size $x$ one trial logs the estimation result after feeding $x$
unique values to the sketch. In addition, it logs the Relative Error (RE) of the estimate, where
$\mathit{RE} = \mathit{MeasuredValue}/\mathit{TrueValue} - 1$. This trial is repeated multiple times (specifically 4K times),
logging all estimation and $\mathit{RE}$ results. The $y$-axis are the lines of the mean and some
quantiles of the distributions of error measured at each $x$-axis point on the graph, including the median. 
This type of graph is called ``pitchfork''.

\subsection{Results}
\label{ssec:results}

\paragraph{Accuracy results}
Our first set of tests run on a 12-core Intel Xeon E5-2620 machine. The accuracy results for the concurrent $\Theta$ sketch
without eager propagation are presented in Figure~\ref{fig:accuracy}. There are two interesting phenomena worth noting.
First, it is interesting to see empirical evaluation reflecting the theoretical analysis presented in Section~\ref{ssec:theta-analysis},
and thus the pitchfork is distorted towards underestimating the number of unique values. Specifically, the mean relative error is smaller
that $0$ (showing a tendency towards underestimating), and the relative error in all measured quantiles tends to be smaller
that the relative error of the sequential implementation.

Second, when the stream size is less than $2k$, $\Theta=1$ and the estimation is the number of values propagated to the
global sketch. Without eager propagation, the number of values in the global sketch depends on the delay in propagation. The
smaller the sketch, the more significant the impact of the delay, and the mean error reaches as high as $94\%$ (the error in
the figure is capped at $10\%$). As the number of propagated values nears $2k$, the delay in propagation is less significant, and
the mean error decreases. Obviously, analytical systems cannot tolerate such high error. The maximum error allowed by
the system is passed as a parameter to the concurrent sketch, and the global sketch uses eager propagation to stay within
the allowed error limit. Figure~\ref{fig:accuracy-adaptive} depicts the accuracy results when applying eager
propagation. The figures are similar when the sketch begins lazy propagation, and the error stays within the $0.04$
limit when eager propagation is used.

\begin{figure*}[tb]
    \setlength{\abovecaptionskip}{0pt}
    \setlength{\belowcaptionskip}{0pt}
    \setlength\textfloatsep{0pt}
    \centering
    \begin{subfigure}{\columnwidth}\centering
    \includegraphics[width=\textwidth]{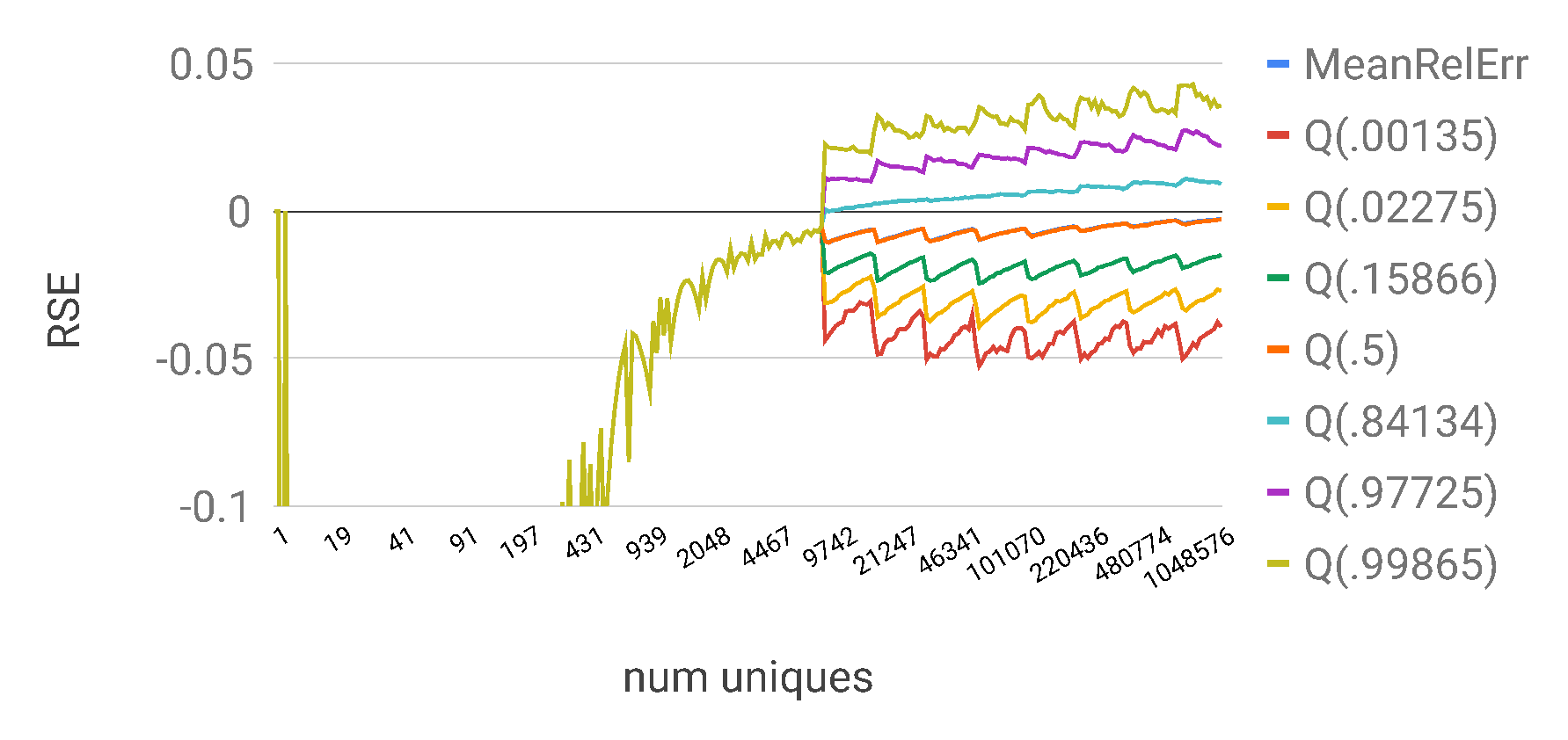}
    \caption{No eager propagation ($e=1.0$)}
    \label{fig:accuracy}
    \end{subfigure}
    \begin{subfigure}{\columnwidth}\centering
    \includegraphics[width=\textwidth]{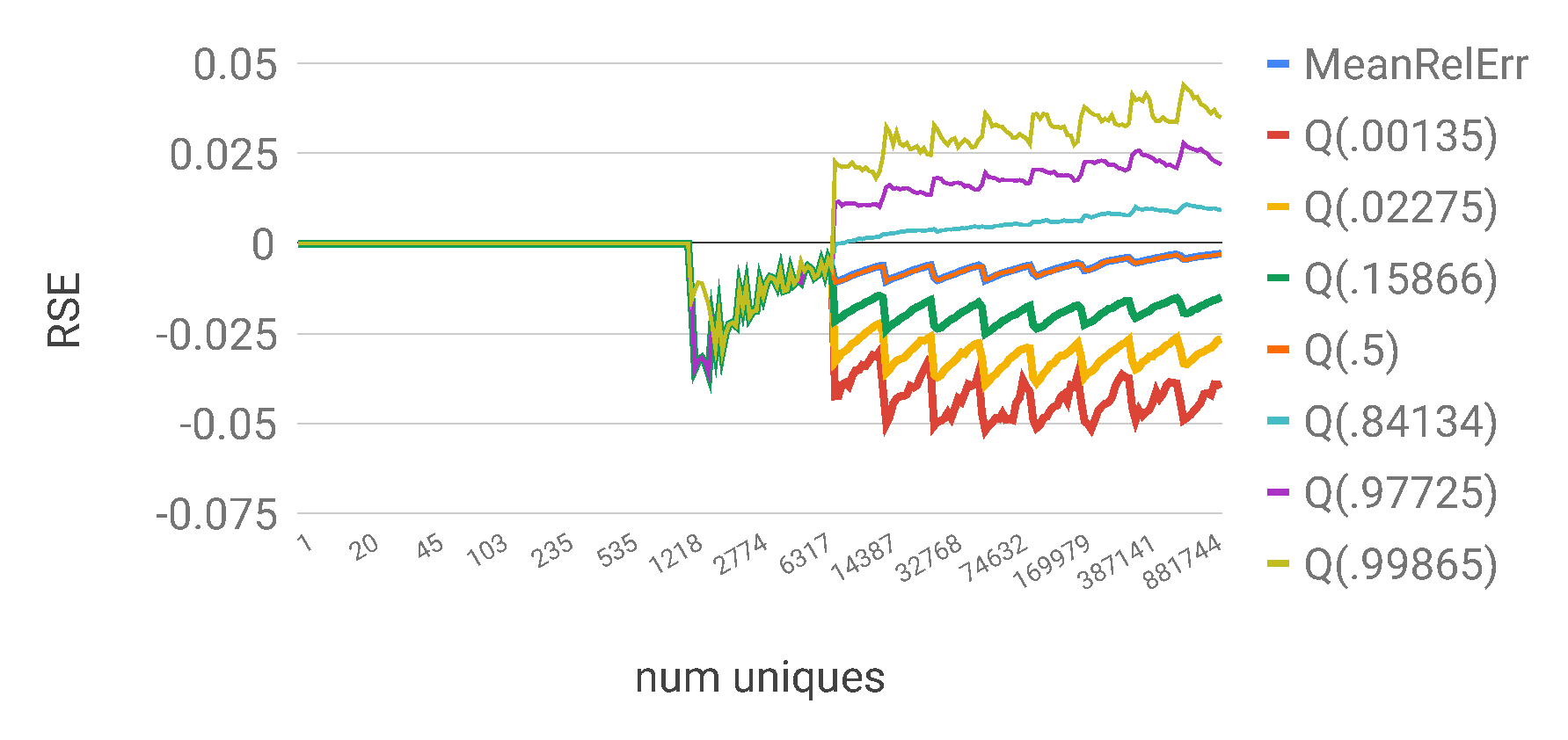}
    \caption{With eager propagation, limit defined by $e=0.04$}
    \label{fig:accuracy-adaptive}
    \end{subfigure}
      \caption{Concurrent Theta Measured Quantiles vs RSE, $k = 4096$.}
      \label{fig:accuracy-res}
\end{figure*}

\paragraph{Write-only workload}
Figure~\ref{fig:throughput:native} presents throughput measurements for a write-only workload. The results are shown in loglog scale.
Figure~\ref{fig:throughput:large} zooms-in on the throughput of large stream sizes.

When considering large stream sizes, the concurrent implementation scales with the number of threads; peaking at
almost $300$M operations per second with $12$ threads. The performance of lock-based implementation, on the other hand,
degrades as the contention on the lock increases with the number of threads. Its peak performance is at
$25$M operations per second with a single thread. Namely, the concurrent $\Theta$ sketch outperforms lock-based implementation
by $12$x, and when comparing the performance of $12$ threads, concurrent implementation outperforms lock-based implementation by more than $45$x. 

For small streams, wrapping a single thread with a lock is the most efficient method. When the stream
contains more than $200$K unique values, using concurrent sketch with $4$ or more local threads is more efficient.
The crossing point of a single local buffer over lock-based implementation is around $700$K unique values.
 
\begin{figure*}[tb]
\setlength{\abovecaptionskip}{0pt}
\setlength{\belowcaptionskip}{0pt}
\setlength\textfloatsep{0pt}
\centering
\begin{subfigure}{1.3\columnwidth}\centering
\includegraphics[width=\textwidth]{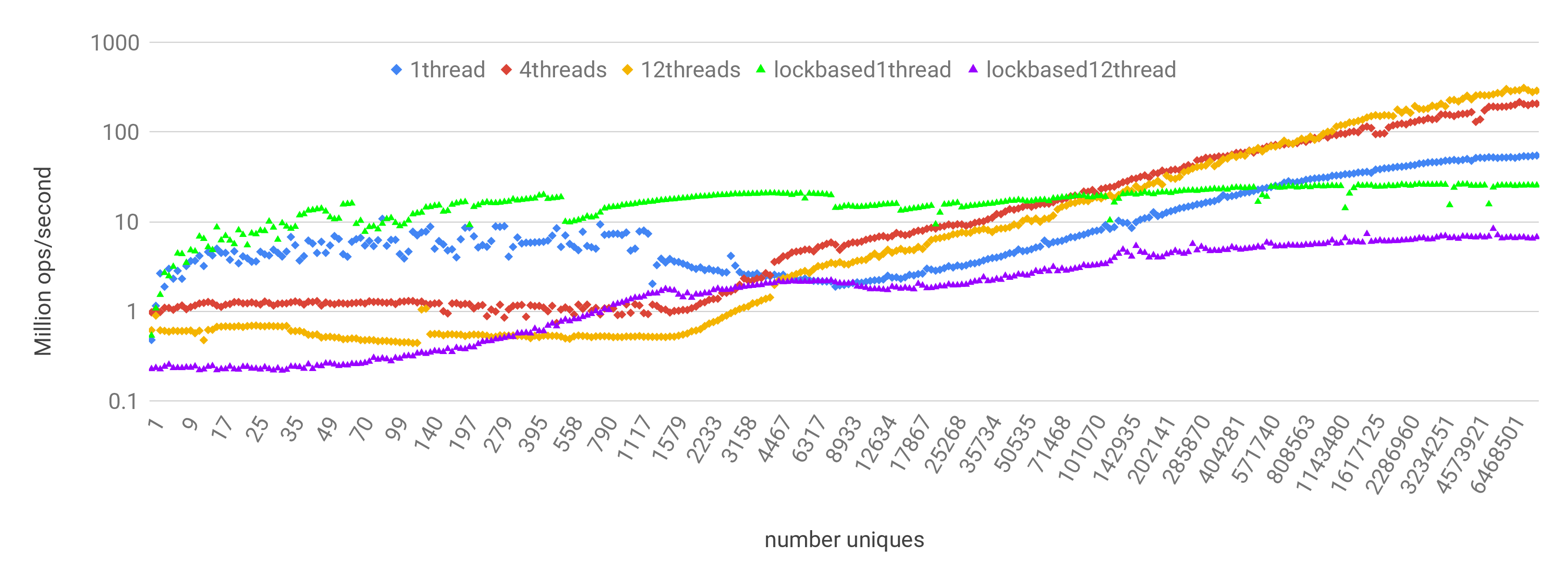}
\caption{Throughput, loglog scale}
\label{fig:throughput:native}
\end{subfigure}
\begin{subfigure}{0.7\columnwidth}\centering
\includegraphics[width=\textwidth]{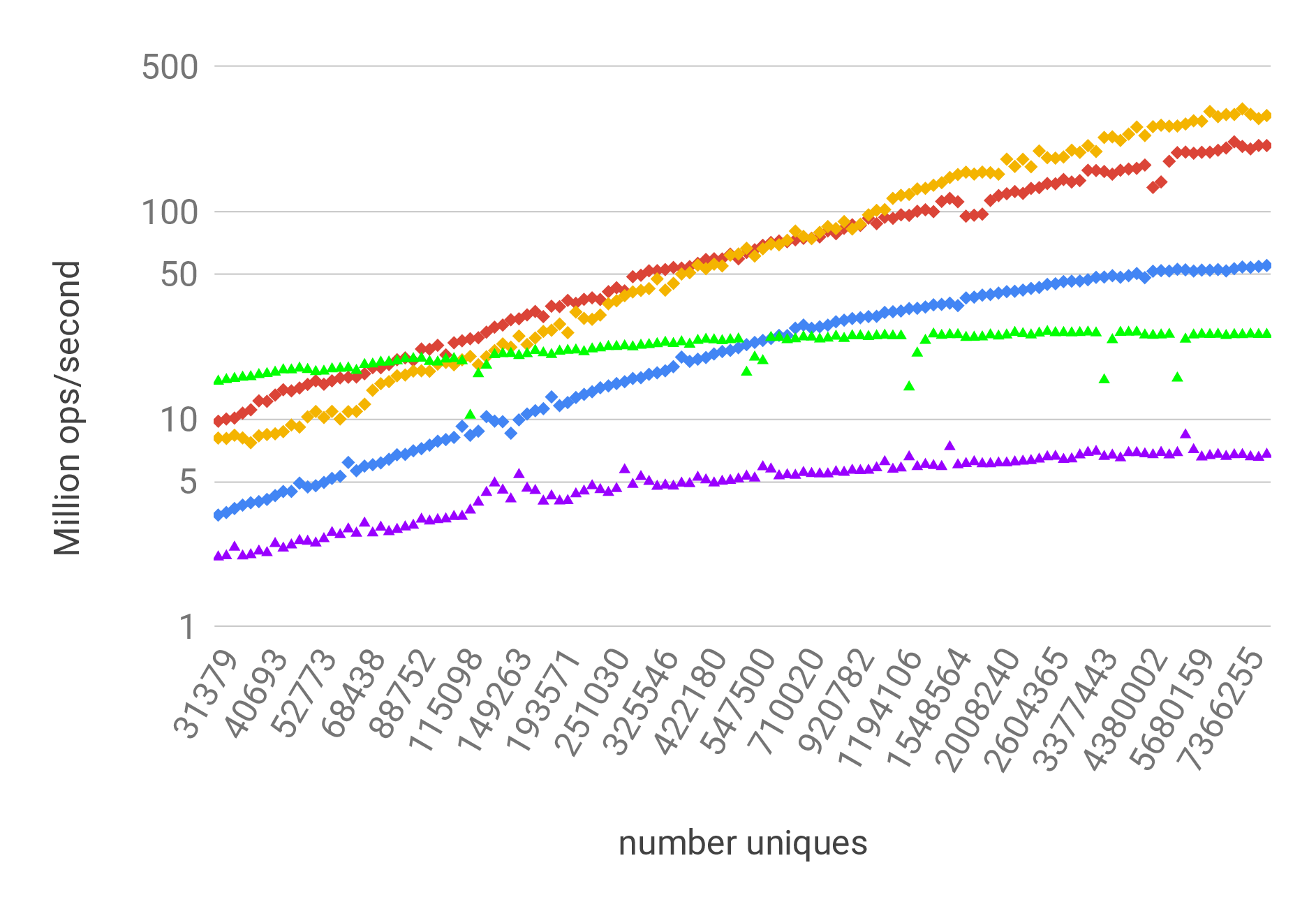}
\caption{Zooming-in on large sketches}
\label{fig:throughput:large}
\end{subfigure}
  \caption{Write-only workload, $k = 4096$, $e=0.04$.}
  \label{fig:throughput}
\end{figure*}

\paragraph{Mixed workload}
Figure~\ref{fig:mixed-throughput} presents throughput measurements
of a mixed read-write workload. We compare runs with a single updating thread and $2$
updating threads (and $10$ background reader threads).
Here we see similar trends as in the write-only workload. However, the effect of
background readers is more pronounced in lock-based implementation than in concurrent $\Theta$ sketch;
this is expected as the reader threads compete for the same lock.
The peak throughput of a single writer-thread in the concurrent implementation is $55$M ops/sec with and
without background readers. The peak throughput of a single writer thread in the lock-based
implementation degrades from $25$M to $23$M operations per second; almost $10$\% slowdown in performance.
Recall that this is when readers are not very frequent, more frequent reads results in even greater
performance degradation.

\begin{figure}[tb]
\setlength{\abovecaptionskip}{0pt}
\setlength{\belowcaptionskip}{0pt}
\setlength\textfloatsep{0pt}
	\centering
	\includegraphics[width=\columnwidth]{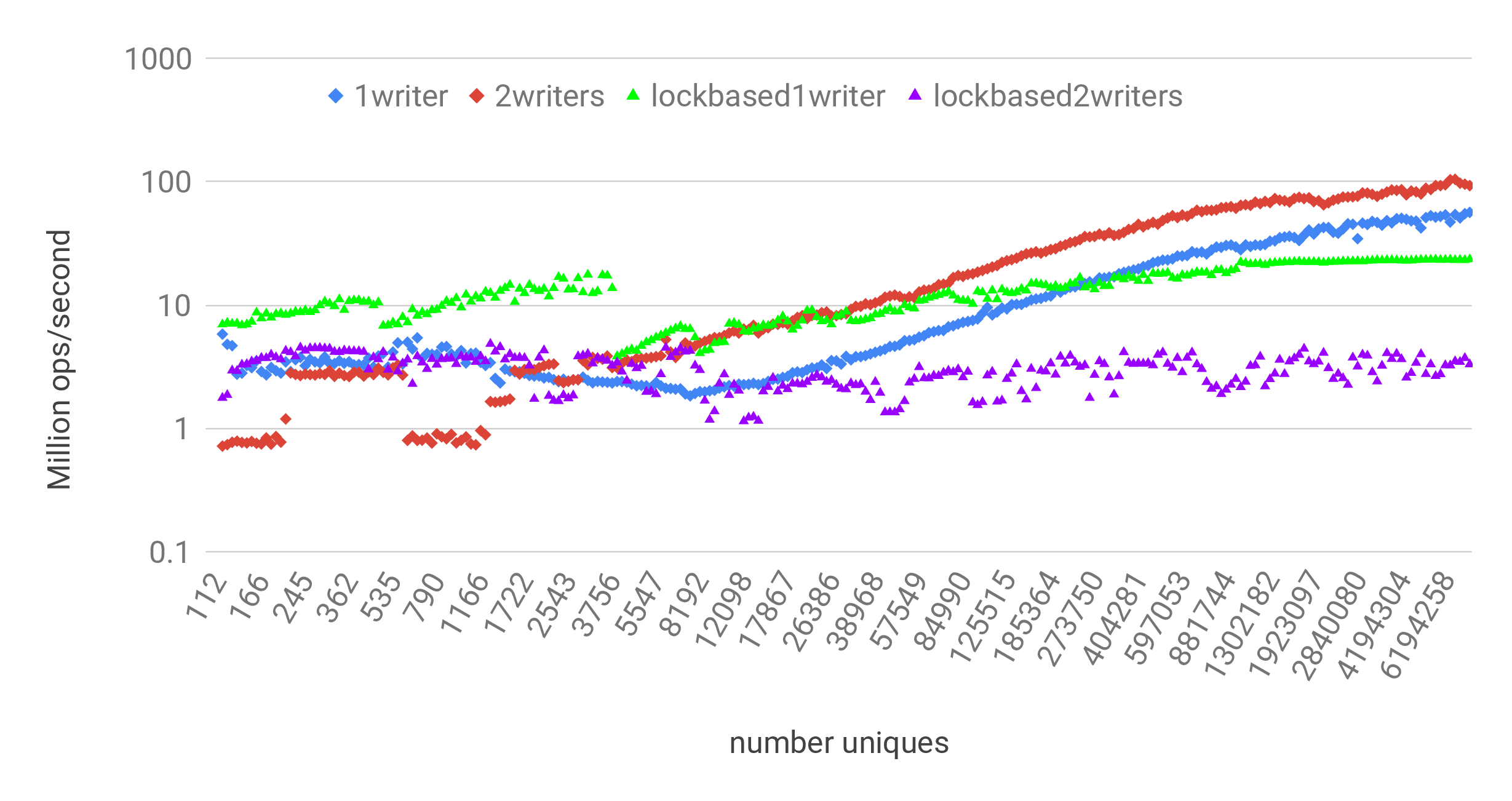}
	\caption{{Mixed workloads: writers with background reads, $k = 4096$, $e=0.04$.}}
	\label{fig:mixed-throughput}
\end{figure}

\paragraph{Scalability results}
To provide a better scalability analysis, we aim to maximize the number of threads working on the
sketch. Therefore, we run this test on a larger machine -- we use a 32-core Xeon E5-4650 processors.
We ran an \emph{update-only} workload in which a sketch is built from a very large stream, repeating
each test 16 times.

In Figure~\ref{fig:performance} (in the introduction) we compare the scalability
of our concurrent $\Theta$ sketch and the original sketch wrapped
with a read/write lock in an update-only workload, for $b=1$ and $k=4096$.
As expected, the lock-based sequential sketch does not scale, and
in fact it performs worse when accessed concurrently by many threads.
In contrast, our sketch achieves almost perfect scalability.
$\Theta$ quickly becomes small enough to allow filtering out most of the updates and so the
local buffers fill up slowly.

\subsection{Accuracy throughput tradeoff}
\label{ssec:tradeoffs}

Eager vs no-eager speedup is presented in Figure~\ref{fig:speedup}. It demonstrates the
speedup of eager propagation over no-eager propagation for small stream sizes, in addition
to the accuracy benefit reported in Figure~\ref{fig:accuracy-res}. 
The improvement goes up to $84$x throughput for tiny sketches, and decreases as the sketch grows.
The slowdown in performance when the sketch size exceed $2k$ can be explained by the reduction
in local buffer size (from $b=16$ to $b=5$) in order to accommodate for the required error bound.

\begin{figure}[tb]
\setlength{\abovecaptionskip}{0pt}
\setlength{\belowcaptionskip}{0pt}
\setlength\textfloatsep{0pt}
	\centering
	\includegraphics[width=\columnwidth]{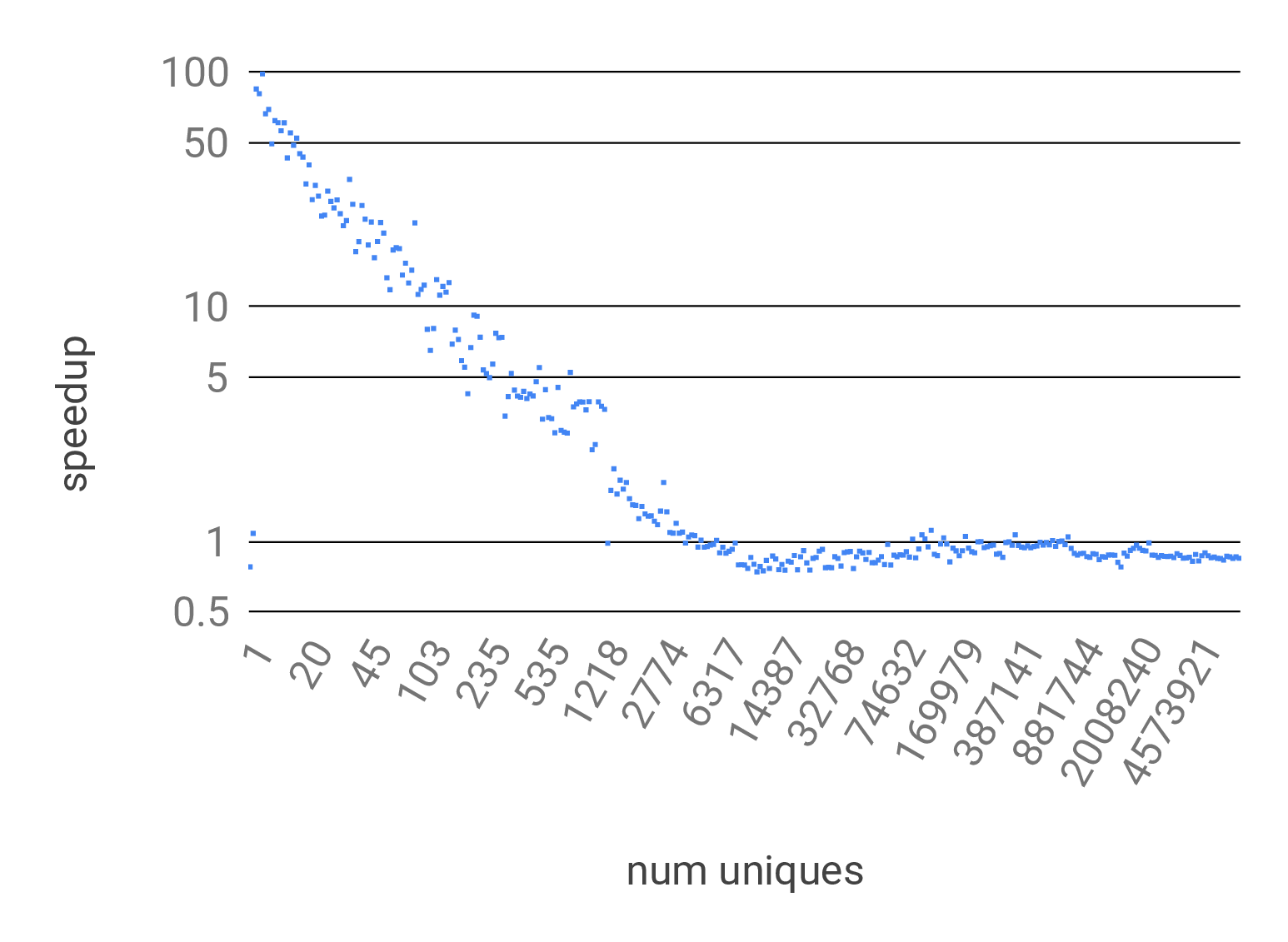}
	\caption{{Throughput speedup of eager ($e=0.04$) vs no-eager ($e=1.0$) propagation, $k = 4096$.}}
	\label{fig:speedup}
\end{figure}

Next we discuss the impact of $k$.
One way to increase the throughput of the concurrent $\Theta$ sketch is by
increasing the size of the global sketch, namely increasing $k$. On the other hand,
this change also increases the error of the estimate.
Table~\ref{tab:tradeoff} presents the tradeoffs between performance and accuracy.
Specifically, it presents the crossing-point, the size of the sketch where concurrent
implementation outperforms lock-based implementation (both running a single thread),
the maximum values (across sketch sizes) of the median error, and 99th percentile error for a variety of $k$ values. 

\begin{table}[htb]
\center{\small{
\begin{tabular}{lrrr}
\hline 
& thpt crossing point & error $Q=0.5$ & error $Q=0.99$ \\
\hline 
$k=256$ &  15,000 &	0.16 & 0.27 \\
\hline 
$k=1024$ &  100,000 &	0.05 & 0.13 \\
\hline 
$k=4096$ & 700,000	& 0.03	& 0.05	\\ 
\hline 
\end{tabular}
}}
\caption{{Performance vs accuracy as a function of $k$}}
\label{tab:tradeoff}
\end{table}

%% file: sections/discussion.tex
\section{Conclusions}
\label{sec:discussion}

Sketches are widely used by a range of applications
to process massive data streams and answer queries about them.
Library functions producing sketches
are optimised to be extremely fast, often digesting tens of millions of stream elements per second. 
We presented a generic algorithm for parallelising such sketches and serving
queries in real-time; the algorithm is strongly linearisable with regards to relaxed semantics.
We showed that the error bounds of two representative sketches,
$\Theta$ and Quantiles, do not increase drastically with such a relaxation. We also
implemented and evaluated the solution, showed it to be scalable {and accurate}, and integrated it into
the open-source Apache DataSketches (Incubating) library. While we analysed only two sketches, future work
may leverage our framework for other sketches. Furthermore, it would be interesting to investigate
additional uses of the hint, for example, in order to dynamically adapt the size of the local buffers
and respective relaxation error.

%% file: sections/artifact.tex
\section{Artifact Appendix}

\subsection{Abstract}

The artifact contains all the JARs of version 0.12 of the DataSketches
library, before it moved into Apache (Incubating), as well as configurations
and shell scripts to run our tests.. It can support the results found in
the evaluated section of our PPoPP'2020 paper Fast Concurrent Data Sketches. To
validate the results, run the test scripts and check the results piped
in the according text output files.

\subsection{Artifact check-list (meta-information)}

{\small
\begin{itemize}
  \item {\bf Algorithm: HLL $\Theta$ Sketch}
  \item {\bf Program: Java code}
  \item {\bf Compilation: JDK 8, and each package is compiled using maven}
  \item {\bf Binary: Java executables}
  \item {\bf Run-time environment: Java}
  \item {\bf Hardware: Ubuntu on 12 core server and 32 core server with hyperthreading disabled}
  \item {\bf Metrics: Throughput and accuracy}
  \item {\bf Output: Runtime throughputs, and runtime accuracy}
  \item {\bf How much time is needed to prepare workflow (approximately)?: Using precomipled packages, none.}
  \item {\bf How much time is needed to complete experiments (approximately)?: 10 hours}
  \item {\bf Publicly available?: Yes.}
  \item {\bf Code licenses (if publicly available)?: Apache License 2.0}
\end{itemize}

\subsection{Description}

\subsubsection{How delivered}

The Apache DataSketches (Incubating) library is an open source project
under Apache License 2.0, and is hosted with code, API specifications,
build instructions, and design documentations on Github. We have provided all the
needed JAR files for running our tests.

\subsubsection{Hardware dependencies}
Our tests require a 12 core server and 32 core machine with hyper-threading disabled

\subsubsection{Software dependencies}
The Apache DataSketches (Incubating) library has been tested on Ubuntu 12.04/14.04,
and is expected to run correctly under other Linux distributions. Building the JAR
files requires JDK 8; the files don't compile otherwise. To use the automated scripts,
we require git and python3 to be installed.

\subsection{Installation}

First, clone the repository:

\begin{framed}

\$ git clone \url{https://github.com/ArikRinberg/FastConcurrentDataSketchesArtifact}

\end{framed}

\noindent We have provided the necessary JAR files for recreating our experiment, and tests
can be done as per Section~\ref{sec:workflow}.

\paragraph{\textbf{Custom compilation.}} 
Alternatively, follow the build instructions on Apache DataSketches (Incubating) apache
page (\url{https://datasketches.apache.org/}), in order to building the above mentioned
JAR files, now called:
\begin{itemize}
  \item incubator-datasketches-java (\url{https://github.com/apache/incubator-datasketches-java})
  \item incubator-datasketches-memory (\url{https://github.com/apache/incubator-datasketches-memory})
  \item incubator-datasketches-characterization (\url{https://github.com/apache/incubator-datasketches-characterization})
\end{itemize}

\noindent The version number of incubator-datasketches-java
and incubator-datasketches-memory must comply with the version numbers required by incubator-datasketches-characterization.
The characterization JAR file is an unsupported open-source code base, and
does not pretend to have the same level of quality as the primary repositories.
These characterization tests are often long running (some can run for days) and very resource intensive, which makes
them unsuitable for including in unit tests. The code in this repository are some of
the test suites we use to create some of the plots on our website and provide evidence for our speed and accuracy claims.
The memory and java releases are provided from Maven Central using the Nexus Repository Manager. Go to 
\url{repository.apache.org} and search for "datasketches".

For convenience we have included these repositories as modules in our main repository along with specific branches and commit id's
that are known to compile. To compile the jar files:
\begin{framed}

\$ git clone \url{https://github.com/ArikRinberg/FastConcurrentDataSketchesArtifact}

\$ cd FastConcurrentDataSketchesArtifact

\$ source customCompile.sh

\end{framed}

\noindent The shell script takes care of initialising the submodules, building the source files, and copying the correct
JAR files to the current directory.

\subsection{Experiment workflow}
\label{sec:workflow}
\paragraph{\textbf{Workflow for precompiled JAR files.}}
For convenience, we provide the JAR files required and the configurations
used to run our tests.

\begin{enumerate}
  \item After cloning the repository:

  \begin{framed}

  \$ cd FastConcurrentDataSketchesArtifact

  \end{framed}

  \noindent In the current working directory, there should either be the following JAR files:

  \begin{itemize}
    \item memory-0.12.1.jar
    \item sketches-core-0.12.1-SNAPSHOT.jar
    \item characterization-0.1.0-SNAPSHOT.jar
  \end{itemize}

  \item Next, run the tests:

  \begin{framed}

  \$ python3 run\_test.py TEST

  \end{framed}

  \noindent Where TEST is one of the following: figure\_1, figure\_6\_a, figure\_6\_b, figure\_7, figure\_8, figure\_9, or table\_2.

  \item \label{n:run-test} The results of each test will be in txt files in the current working directory, either SpeedProfile or AccuracyProfile:
  \paragraph{\textbf{SpeedProfile:}} The txt file contains three columns: \textbf{InU} -- the number of unique items (the $x$
  axis of most graphs), \textbf{Trials} -- the number of trials for this run, \textbf{nS/u} -- nano seconds per update. The $y$
  axis of the throughput graphs is given as updates per second, therefore a conversion is needed.
  \paragraph{\textbf{AccuracyProfile:}} The txt files contains the columns corresponding to the name on the figure, where
  \textbf{InU} is the number of unique items.

\end{enumerate}

\paragraph{\textbf{Workflow for custom JAR files.}}

\begin{enumerate}
  \item After cloning the repository:

  \begin{framed}

  \$ cd FastConcurrentDataSketchesArtifact

  \end{framed}

  \noindent In the current working directory, there should either be the following JAR files:

  \begin{itemize}
    \item datasketches-memory-1.1.0-incubating.jar
    \item datasketches-java-1.1.0-incubating.jar
    \item datasketches-characterization-1.0.0-incubating-SNAPSHOT.jar
  \end{itemize}

  \item For each .conf file in the conf\_files folder, the following line must be altered:
  
  \noindent \textbf{From:} JobProfile=com.yahoo.sketches.characterization.uniquecount.TEST

  \noindent \textbf{To:} JobProfile=org.apache.datasketches.characterization.theta.concurrent.TEST

  \noindent Where TEST is either ConcurrentThetaAccuracyProfile or ConcurrentThetaMultithreadedSpeedProfile.

  \item Finally, the following line must be altered in run\_test.py:
  
  \noindent \textbf{From:} CMD = 'java -cp "./*" com.yahoo.sketches.characterization.Job {}'

  \noindent \textbf{To:} CMD = 'java -cp "./*" org.apache.datasketches.Job {}'

  \item The tests can now be run as explained in Item~\ref{n:run-test}.

\end{enumerate}

\subsection{Evaluation and expected result}

For Figures 1, 7, 8 and 9, the expected results are runtime throughput in nanoseconds
per update. These figures show updates per second, therefore a conversion is needed;
if the result is $x$, then the data-point is $1000/x$. For Figure 6, the
expected results are accuracy.

\subsection{Experiment customization}

Each curve in each experiment is customised in the corresponding configure file.
The main customisations for the conf files are:
\begin{itemize}
  \item \textbf{Trials\_lgMinU / Trials\_lgMaxU:} Range of number of unique numbers over which to run the test.
  \item \textbf{LgK:} Size of the global sketch.
  \item \textbf{CONCURRENT\_THETA\_localLgK:} Size of the local sketch.
  \item \textbf{CONCURRENT\_THETA\_maxConcurrencyError:} Maximum error due to concurrency. For non-eager tests, set to 1.
  \item \textbf{CONCURRENT\_THETA\_numWriters:} Number of writer threads.
  \item \textbf{CONCURRENT\_THETA\_ThreadSafe:} IS true if the test should use the concurrent implementation,
  false if the test should use a lock-based implementation.
\end{itemize}


\subsection{Methodology}

Submission, reviewing and badging methodology:

\begin{itemize}
  \item \url{http://cTuning.org/ae/submission-20190109.html}
  \item \url{http://cTuning.org/ae/reviewing-20190109.html}
  \item \url{https://www.acm.org/publications/policies/artifact-review-badging}
\end{itemize}